\documentclass[aos,preprint]{imsart}

\RequirePackage[OT1]{fontenc}
\RequirePackage{amsthm,amsmath}
\RequirePackage[numbers]{natbib}
\RequirePackage[colorlinks,citecolor=blue,urlcolor=blue]{hyperref}

\startlocaldefs
\numberwithin{equation}{section}
\theoremstyle{plain}
\newtheorem{thm}{Theorem}[section]

\usepackage{amsfonts}
\usepackage{graphicx}   % need for figures

\endlocaldefs

\begin{document}

\newcommand{\aalpha}{\boldsymbol\alpha}
\newcommand{\bbeta}{\boldsymbol\beta}
\newcommand{\ggamma}{\boldsymbol\gamma}
\newcommand{\GGamma}{\boldsymbol\Gamma}
\newcommand{\nnabla}{\boldsymbol\nabla}
\newcommand{\R}{\mathbb{R}}
\newcommand{\xx}{\mathbf{x}}
\newcommand{\XX}{\mathbf{X}}
\newcommand{\llog}{\mathrm{log}}
\newcommand{\SSigma}{\boldsymbol\Sigma}
\newcommand{\mmu}{\boldsymbol\mu}
\newcommand{\N}{\mathcal{N}}
\newcommand{\yy}{\mathbf{y}}
\newcommand{\zz}{\mathbf{z}}
\newcommand{\ZZ}{\mathbf{Z}}
\newcommand{\WW}{\mathbf{W}}
\newcommand{\pp}{\mathbf{p}}
\newcommand{\PP}{\mathrm{P}}
\newcommand{\qq}{\mathbf{q}}
\newcommand{\GG}{\mathbf{G}}
\newcommand{\g}{\mathbf{g}}
\newcommand{\HH}{\mathbf{H}}
\newcommand{\A}{\mathbf{A}}
\newcommand{\fee}{\mathrm{FEE}}
\newcommand{\feens}{\mathrm{FEE{-}NS}}
\newcommand{\feees}{\mathrm{FEE{-}ES}}

\begin{frontmatter}

% "Title of the paper"
\title{Metropolis-Hastings Sampling Using Multivariate Gaussian Tangents}
\runtitle{MH-MGT Sampling}

\begin{aug}
\author{\fnms{Alireza S.} \snm{Mahani}\thanksref{m1}\corref{}\ead[label=e1]{alireza.mahani@sentrana.com}}
\address{\printead{e1}}
\and
\author{\fnms{Mansour T.A.} \snm{Sharabiani}\thanksref{m2,m3}\ead[label=e2]{???}}
%\address{\printead{e2}}
\affiliation{Scientific Computing Group, Sentrana Inc., USA\thanksmark{m1} and Department of Epidemiology and Biostatistics, Imperial College, UK \thanksmark{m2} and Daric Solutions LLC, USA \thanksref{m3}}
%\printead{e1}\\
%\phantom{E-mail:\ }\printead*{e2}}

\runauthor{Mahani and Sharabiani}
\end{aug}

\begin{abstract}
We present MH-MGT, a multivariate technique for sampling from twice-differentiable, log-concave probability density functions. MH-MGT is Metropolis-Hastings sampling using asymmetric, multivariate Gaussian proposal functions constructed from Taylor-series expansion of the log-density function. The mean of the Gaussian proposal function represents the full Newton step, and thus MH-MGT is the stochastic counterpart to Newton optimization. Convergence analysis shows that MH-MGT is well suited for sampling from computationally-expensive log-densities with contributions from many independent observations. We apply the technique to Gibbs sampling analysis of a Hierarchical Bayesian marketing effectiveness model built for a large US foodservice distributor. Compared to univariate slice sampling, MH-MGT shows 6x improvement in sampling efficiency, measured in terms of `function evaluation equivalents per independent sample'. To facilitate wide applicability of MH-MGT to statistical models, we prove that log-concavity of a twice-differentiable distribution is invariant with respect to 'linear-projection' transformations including, but not restricted to, generalized linear models.
\end{abstract}



\end{frontmatter}

% AOS,AOAS: If there are supplements please fill:
%\begin{supplement}[id=suppA]
%  \sname{Supplement A}
%  \stitle{Title}
%  \slink[doi]{10.1214/00-AOASXXXXSUPP}
%  \sdatatype{.pdf}" 
%  \sdescription{Some text}
%\end{supplement}

\section{Introduction}
Univariate samplers can be applied to multivariate distributions using the Gibbs sampling framework. In high-dimensional state spaces, univariate samplers become inefficient as they require many function evaluations for each Gibbs cycle. The problem is more pronounced when the sampled distribution has a significant correlation structure. Multivariate samplers such as the shrinking rank slice sampler \cite{bib:mt-shrink-rank-slicer} or adaptive Metropolis-Hastings sampler \cite{bib:rr-adaptive-mh} can be more efficient in such circumstances, but at the expense of greater need for tuning. Finding optimal parameters for multivariate samplers can create significant manual work, which is perhaps why Bayesian inference software such as OpenBUGS\footnote{\url{http://www.openbugs.info/w/}} or JAGS\footnote{\url{http://mcmc-jags.sourceforge.net/}} primarily use univariate techniques such as slice sampler with stepout \cite{bib:neal} or adaptive rejection sampler \cite{bib:gilks-wild} for non-standard distributions. In this paper, we seek to develop a new multivariate sampling technique that is both efficient (i.e. requires few \underline{F}unction \underline{E}valuations \underline{E}quivalents (or FEE's) per independent sample) and robust (i.e. requires little tuning).

Metropolis-Hastings using Multivariate Gaussian Tangents, MH-MGT, is a technique for sampling from twice-differentiable, log-concave probability density functions. It uses asymmetric, multivariate Gaussians constructed from Taylor-series expansion of the log-density as proposal function. The mean of the Gaussian proposal function represents the full Newton step, and thus MH-MGT is the stochastic counterpart to Newton optimization. MH-MGT involves only 2 function/gradient/Hessian evaluations per sample, and its non-local jumps can lead to low autocorrelation. The algorithm, and its convergence and mixing properties, is described in detail in Section \ref{sec:mh-mgt}.

In Section \ref{sec:case} we apply MH-MGT to a Hierarchical Bayesian logistic regression problem, used in a marketing effectiveness study for a large US foodservice distributor. We observe that, compared to univariate slice sampler, MH-MGT is $\sim$6x more `efficient'. Sampler efficiency is formally defined as the number of function evaluation equivalents per independent sample, factoring in both the computational burden of samples and their autocorrelation \cite{bib:mt-thesis}. In order to achieve better convergence and mixing for MH-MGT, we make two adjustments: We use the first half of burn-in iterations to run MH-MGT in non-stochastic mode, i.e. accepting full Newton step rather than drawing from the proposal function and applying the MH rejection test. We also partition the 50-dimensional state space for low-level coefficients into 10 groups of 5 and apply Gibbs sampling to the resulting partitions. This not only counters the `curse of dimensionality', but also optimizes computational burden per sample given the quadratic scaling of Hessian matrix calculation with its size. To perform a fair comparison of MH-MGT against several other sampling techniques, we carefully measure and report their sampling efficiency when applied to a logistic regression log-likelihood function.

Limitations of MH-MGT and future research are discussed in Section \ref{sec:discussion}, including ways to improve convergence and mixing of MH-MGT, potential for faster Hessian calculation including parallelization, handling boundary conditions, and extending MH-MGT beyond log-concave densities.

Proving negative definiteness of the Hessian (prerequisite for current version of MH-MGT) can be challenging in a high-dimensional state space. In Appendix \ref{sec:theorem} we prove that a log-density with a negative-definite Hessian retains this property if one or more of its parameters undergo a linear-projection transformation. Such transformations are very common in statistical modeling (including generalized linear regression models) and therefore this theorem reduces the problem of proving negative-definiteness of Hessian to a much smaller dimensionality. Given that many common distributions are log-concave \cite{bib:gilks-wild}, MH-MGT can be applied to a wide array of problems in statistical modeling.

\section{MH-MGT Sampling} \label{sec:mh-mgt}
We set the stage by providing a brief overview of Metropolis-Hastings sampling algorithm.
\subsection{Overview of Metropolis-Hastings MCMC}
In Metropolis-Hastings (MH) MCMC sampling of probability distribution $p(\zz)$ (\cite{bib:mh}), we use a proposal function $q(\zz|\zz^{\tau})$ to generate a new sample $\zz^*$ and accept it with probability $A_k(\zz^*,\zz^{\tau})$ where % \footnote{We adopt the mathematical notation of \cite{bib:bishop} throughout this manuscript.}
\begin{equation}\label{eq:mh-A}
A(\zz^*,\zz^{\tau}) = \mathrm{min}(1,\frac{p(\zz^*)q(\zz^{\tau}|\zz^*)}{p(\zz^{\tau})q(\zz^*|\zz^{\tau})})
\end{equation}
The transition probability $q(\zz|\zz')A(\zz',\zz)$ satisfies detailed balance:
\begin{equation}
\begin{array}{lcl}
p(\zz)q(\zz|\zz')A(\zz',\zz) &=& \mathrm{min}(p(\zz)q(\zz|\zz'),p(\zz')q(\zz'|\zz)) \\
&=& \mathrm{min}(p(\zz')q(\zz'|\zz),p(\zz)q(\zz|\zz')) \\
&=& p(\zz')q(\zz'|\zz)A(\zz,\zz')
\end{array}
\end{equation}
The detailed balance property ensures that $p(\zz)$ is invariant under MH transitions.
\subsection{MGT Proposal Function}
MH-MGT algorithm uses as proposal function a multivariate Gaussian fitted locally to the distribution being sampled. This Gaussian fit is based on the following Taylor's series expansion of a multivariate scalar function:
\begin{equation}
\label{taylor}
f(\xx) \approx f(\xx_0) + \g(\xx_0)^T . (\xx-\xx_0) + \frac{1}{2} (\xx-\xx_0)^T \HH(\xx_0) (\xx-\xx_0)
\end{equation}
where $f:\R^K \rightarrow \R$, and $\g$ and $\HH$ stand for the gradient vector and Hessian matrix for $f$, respectively. If we assume that $f$ represents the logarithm of a concave probability distribution function (PDF), then the above approximation is equivalent to fitting the PDF (which we call $F$) with a multivariate Gaussian:
\begin{equation}
\label{gauss}
F(\xx) = \frac{1}{(2\pi)^{K/2}|\SSigma|^{1/2}} e^{-\frac{1}{2}(\xx-\mmu)^T \SSigma^{-1}(\xx-\mmu)}
\end{equation}
From comparing \eqref{taylor} and \eqref{gauss} it is obvious that the precision matrix is the same as the negative Hessian: $\SSigma^{-1}=-\HH(\xx_0)$. To find the mean of the fitted Gaussian, we observe that Gaussian mean maximizes the PDF (and its log). Therefore, finding the mean is equivalent to maximizing \eqref{taylor}, i.e. setting its derivative with respect to $\bbeta$ to zero. After some calculus, we arrive at:
\begin{equation} \label{eq:newton-step}
\mmu = \xx_0 - \HH^{-1}(\xx_0) \g(\xx_0)
\end{equation}
Therefore, our proposal function $q(.|\xx)$ is formally defined as:
\begin{equation} \label{eq:proposal}
q(.|\xx) = \N(\xx - \HH^{-1}(\xx) \g(\xx),-\HH^{-1}(\xx))
\end{equation}
Note that Eq. \eqref{eq:newton-step} is simply the full Newton step \cite{bib:nocedal}. We can therefore think of MH-MGT as the stochastic counterpart of Newton optimization. In optimization, we select the mean of the fitted Gaussian as the next step, while in MH-MGT we draw a sample from the fitted Gaussian and apply MH test to accept or reject it. Also, note that in the special case where the sampled PDF is Gaussian, $f(\xx)$ is quadratic and therefore the proposal function is identical to the sampled PDF. In this case $A(\zz',\zz)$ is always equal to $1$, implying an acceptance rate of $100\%$. Finally, note that for the fitted Gaussian to exist its covariance matrix must be positive definite. From \eqref{eq:proposal} we see that this is equivalent to the Hessian of the log-density (and its inverse) to be negative definite.
\subsection{Full MH-MGT Algorithm}
Combining the last two sections, the following steps describe the MH-MGT algorithm for drawing a sample $\xx_{new}$ from a PDF with log-density $f(\xx)$, given last sample $\xx_{old}$:
\begin{enumerate}
\item
Evaluate the log-density function and its gradient and Hessian at $\xx_\text{old}$: $f_\text{old},\g_\text{old},\HH_\text{old}$.
\item
Construct the multivariate Gaussian proposal function at $q(.|\xx_{old})$ using Eq. \eqref{eq:proposal} and $\xx=\xx_{old}$.
\item
Draw a sample $\xx_{prop}$ from $q(.|\xx_{old})$, and evaluate $logq_{prop}=\log(q(\xx_{prop}|\xx_{old}))$.
\item
Evaluate the log-density function and its gradient and Hessian at $\xx_{prop}$: $f_{prop},\g_{prop},\HH_{prop}$.
\item
Construct the multivariate Gaussian proposal function at $q(.|\xx_{prop})$ using Eq. \eqref{eq:proposal} and $\xx=\xx_{prop}$, and evaluate $logq_{old}=\log(q(\xx_{old}|\xx_{prop}))$.
\item
Calculate the ratio $r=\exp((f_{prop}-f_{old})+(logq_{old}-logq_{prop}))$.
\item
If $r \geq 1$ accept $\xx_{prop}$: $\xx_{new} \leftarrow \xx_{prop}$. Else, draw a uniform random deviate $s$ from $[0,1)$. If $s<r$, then accept $\xx_{prop}$: $\xx_{new} \leftarrow \xx_{prop}$, else reject $\xx_{prop}$: $\xx_{new} \leftarrow \xx_{old}$.
\end{enumerate}
Figure \ref{fig:mh-mgt} illustrates the MH-MGT algorithm graphically.

\begin{figure}
\vspace{6pc}
\includegraphics[scale=0.75]{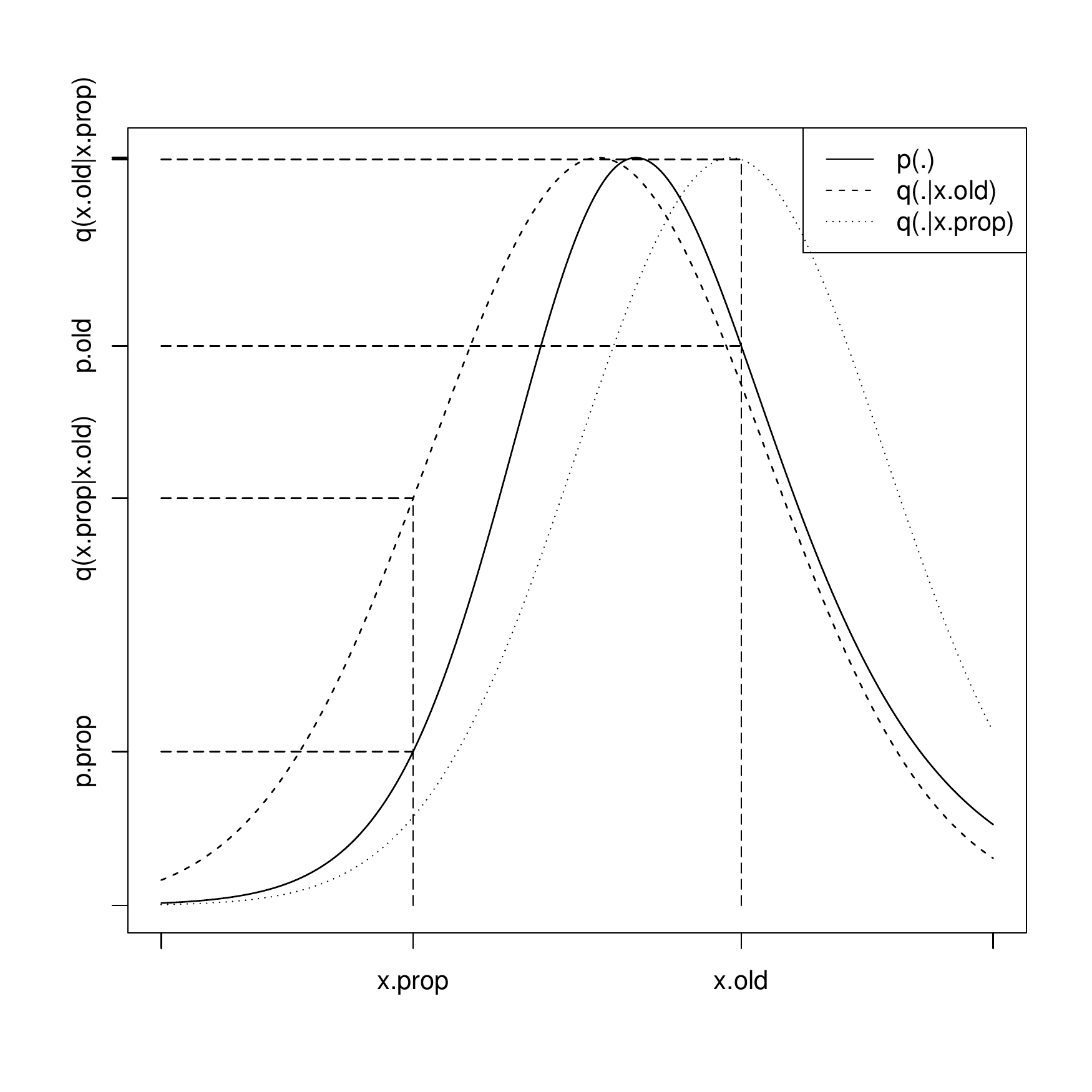}
\caption[]{Illustration of MH-MGT algorithm: The Gaussian proposal function $q(.|\xx_{old})$ is fit to $p(.)$ at $\xx_{old}$. A sample $\xx_{prop}$ is drawn from this proosal function, and a Gaussian $q(.|\xx_{prop})$ is fit to $p(.)$ at $\xx_{prop}$. This is followed by the MH rejection test by forming the ratio $(p(\xx_{prop})q(\xx_{old}|\xx_{prop}))/(p(\xx_{old})q(\xx_{prop}|\xx_{old}))$}.
\label{fig:mh-mgt}
\end{figure}

\subsection{Convergence and Mixing}
For Gaussian distributions, MH-MGT leads to 100\% acceptance and exact sampling since the proposal functions remain the same, and identical to the sampled distribution. In general, log-density has a non-zero third derivative, meaning that the fitted Gaussians at different locations deviate from the actual distribution and from each other, leading to non-zero rejection rate and auto-correlated MCMC chains. Figure \ref{fig:bern-converge} illustrates the behavior of MH-MGT for sampling from likelihood function for the parameter of a Poisson distribution with a single observation $\{2\}$. As the initial point is moved father away from the distribution mode, it takes longer for the chain to converge. Even after convergence, we encounter occasional long jumps, each followed by an extended rejection period. A simple yet effective way to facilitate convergence is to perform the first few iterations in a non-stochastic mode where, instead of MH sampling, we simply take Newton steps. This allows MH-MGT have the same convergence behavior as Newton optimization. (See Figure \ref{fig:bern-converge}.)

\begin{figure}
\vspace{6pc}
\includegraphics[scale=0.75]{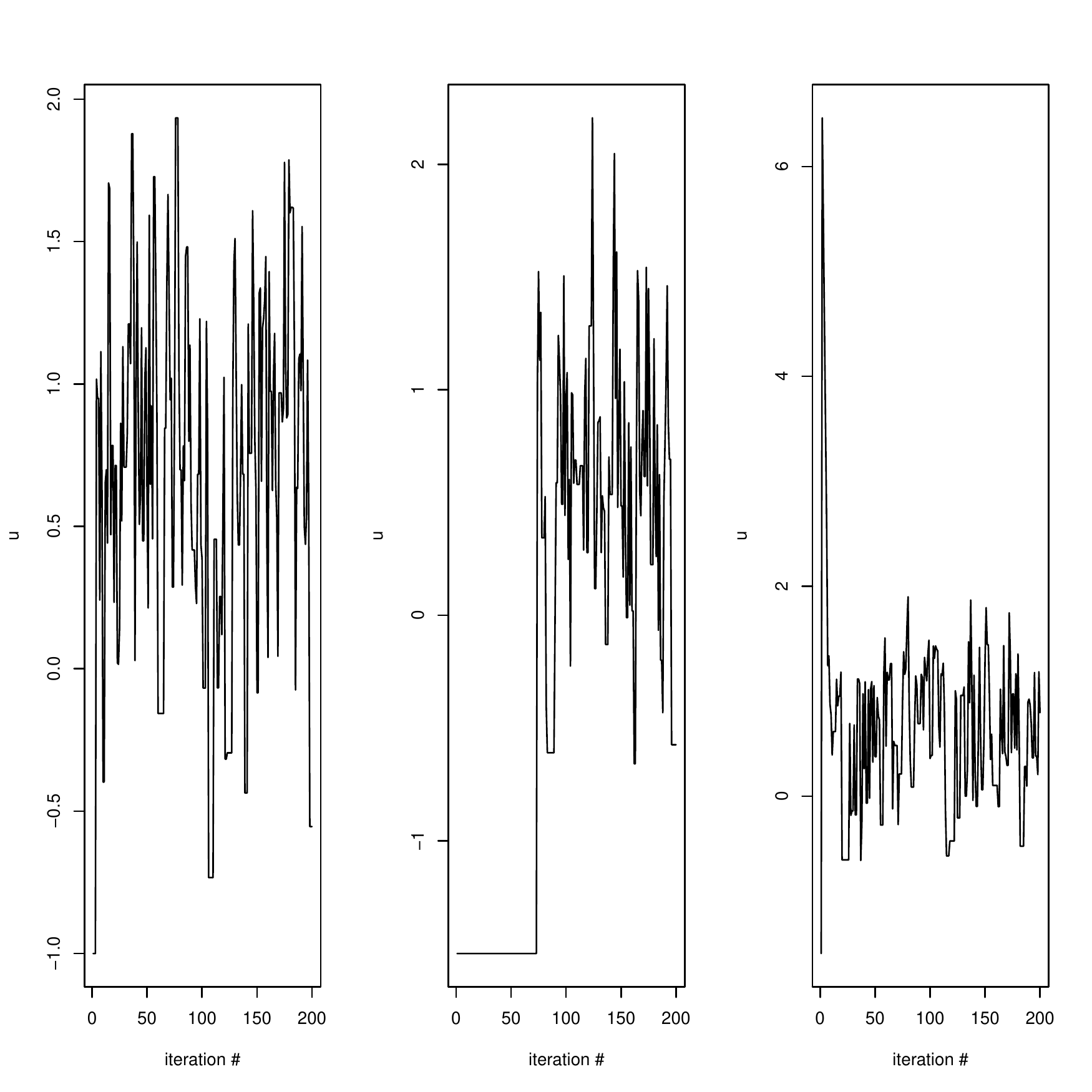}
\caption[]{MH-MGT convergence and mixing for various initial points for Poisson likelihood function (as a function of distribution parameter, given a single observation $\{2\}$). Left: Sample MCMC chain using MH-MGT, with starting point of $u=-1.0$. Middle: Increasing distance of initial point from distrubtion mode ($log(2) \simeq 0.69$) to $u=-1.5$ has a visible negative effect on convergence. Right: Keeping $u=-1.5$, if we use Newton step in the first 5 iteraions, convergence is improved significantly. In all cases, large deviations (towards a negative value) from distribution mode are often associated with an extended period of rejection before the chain returns to higher-density areas.}
\label{fig:bern-converge}
\end{figure}

Efficient mixing in MH-MGT depends on absence of large changes to the mean and std of the proposal function (relative to the std calculated at mode) as the chain moves within a few std's of the mode. To quantify this, we begin with the Taylor series expansion of log-density around its mode $\mu_0$ (where $f'(\mu_0)=0$), this time keeping the third-order term as well (and restricting our analysis to univariate case):
\begin{eqnarray}
f(x) &\approx& f(\mu_0) - \frac{1}{2}\tau_0(x-\mu_0)^2 + \frac{1}{6}\kappa_0(x-\mu_0)^3, \\
\tau_0 &\equiv& -f''(\mu_0) \\
\kappa_0 &\equiv& f'''(\mu_0)
\end{eqnarray}
Applying the Newton step to the above formula, we can arrive at mean $\mu(x)$ and precision $\tau(x)$ of the fitted Gaussian at $x$:
\begin{subequations}
\begin{eqnarray}
\mu(x) &=& x - \frac{f'(x)}{f''(x)} \\
&=& x - \frac{-\tau_0 (x-\mu_0) + \frac{1}{2} \kappa (x-\mu_0)^2}{-\tau_0 + \kappa_0 (x-\mu_0)} \\
\tau(x) &=& \tau_0 + \kappa_0 (x-\mu_0)
\end{eqnarray}
\end{subequations}
We now form the following dimensionless ratios, which we need to be much smaller than 1 in order to have good mixing for MH-MGT:
\begin{eqnarray}
|\mu(x)-\mu_0| \: . \: \tau_0 \ll 1 \\
|\tau(x)-\tau_0| / \tau_0 \ll 1
\end{eqnarray}
when $|x-\mu_0| \: . \: \tau_0 \sim O(1)$. Some algebra shows that both these conditions are equivalent to:
\begin{equation} \label{eq:mixing-condition}
\eta_0 \equiv |\kappa_0| \: . \: \tau_0^{-3/2} \ll 1
\end{equation}
In the example of Figure \ref{fig:bern-converge}, we can calculate $\eta_0 \simeq 0.71$, which suggests less-than-excellent mixing according to our rule-of-thumb.

The Poisson distribution example, however, is rather contrived: In real-world applications, we are interested in efficient sampling from computationally-expensive log-densities. In particular, statistical models often involve many independent observations, each contributing an additive term to the log-likelihood function:
\begin{equation}
f(x) = \sum_{i=1}^N f_i(x)
\end{equation}
We can therefore re-write equation \eqref{eq:mixing-condition} as:
\begin{equation}
\eta_0 = |\sum_{i=1}^N f'''_i(x)| \: . \: [\sum_{i=1}^N f''_i(x)]^{-3/2}
\end{equation}
Assuming that individual terms in above equation remain bounded, we can easily see that
\begin{equation}
\eta_0 \sim O(N^{-1/2})
\end{equation}
We therefore arrive at the following rule of thumb: {\em MH-MGT sampling becomes more efficient as the number of observations in a log-likelihood function increases}. Figure \ref{fig:mixing-N} illustrates the impact of $N$ on mixing of MH-MGT for the Poisson distribution example.
\begin{figure}
\vspace{6pc}
\includegraphics[scale=0.75]{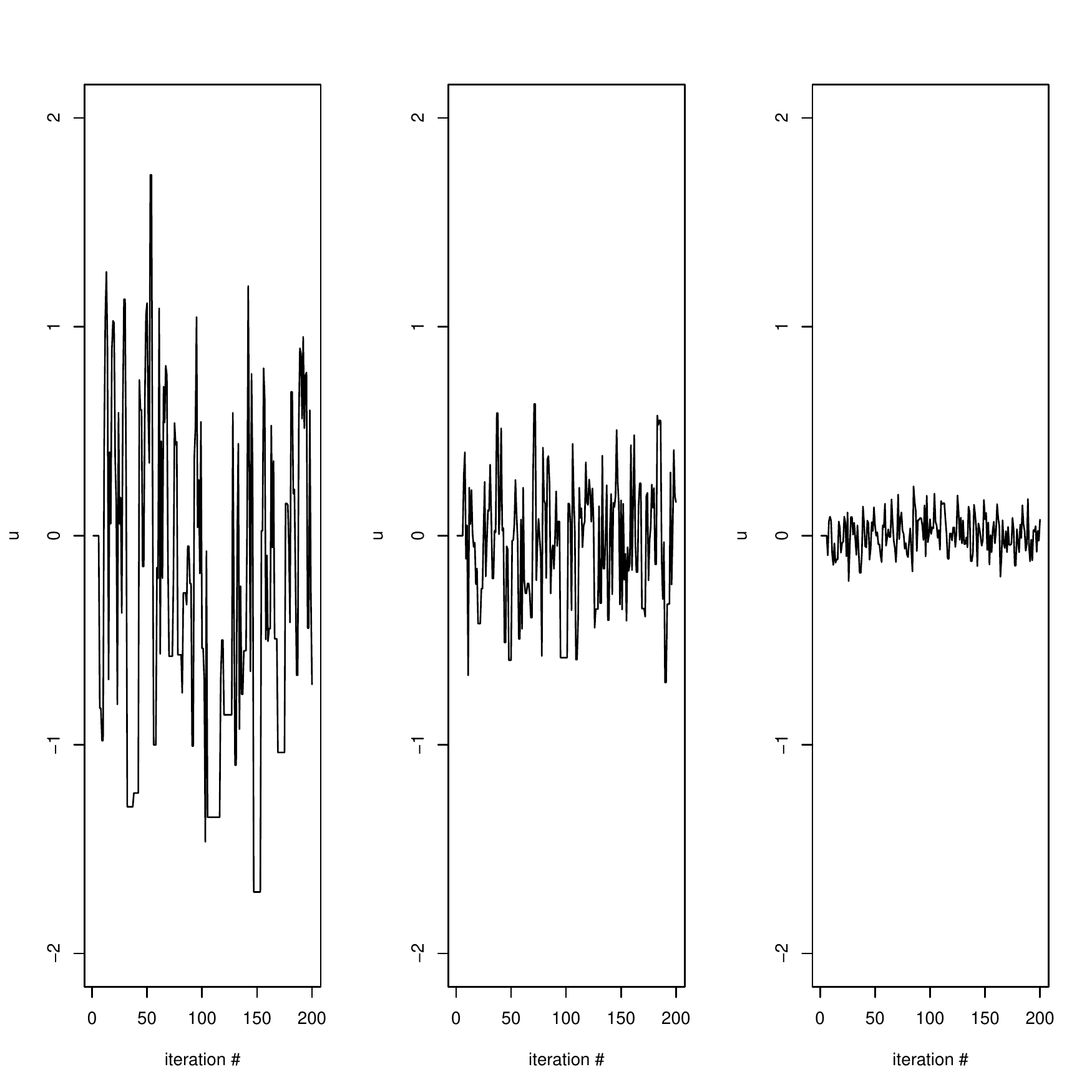}
\caption[]{Impact of number of observations on mixing for MH-MGT sampling of the likelihood function for (parameter of) Poisson distribution, with a single observation of $\{1\}$. Left: 1 observation. Middle: 10 observations. Right: 100 observations. With increased observations, }
\label{fig:mixing-N}
\end{figure}

\subsection{MH-MGT for High-Dimensional Problems}
In order to apply MH-MGT to twice-diffenrentiable high-dimensional probability density functions, we must first prove that Hessian is negative definite. To this end, we have proven an invariance theorem in Appendix \ref{sec:theorem} that significantly simplifies this task when the function results from linear-projection changes of variable applied to a base distribution. This theorem covers what is known as the generalized linear family of models \cite{bib:glm}, but is more general.

The primary motivation behind MH-MGT is to amortize the cost of function and derivative evaluations over multiple parameters, thereby increasing sampling efficiency. Sampling efficiency of MH-MGT does improve with dimensionality of parameter space, but up to an extent. There are two factors that act as countering forces to create an optimal dimensionality: 1) quadratic scaling of time needed for Hessian calculation with dimensionality, 2) curse of dimensionality, i.e. decreased convergence and mixing of MH-MGT chain as dimensionality increases. Our experiments show that MH-MGT tends to work best when we partition the parameter space into 5-10 dimensional chunks, and apply the algorithm to each chunk using Gibbs sampling. This strategy is used in Section \ref{sec:case}.

\section{Application to Quantitative Marketing} \label{sec:case}
We illustrate the performance of MH-MGT by applying it to a problem in quantitative marketing.
\subsection{Model}
A large US foodservice distributor makes promotional offers of select products to its customers (restaurant operators) across the US. The marketing team wants to understand the drivers of customer response (acceptance vs. rejection of offer) so that it can focus limited funds and customer attention on high-opportunity offers to maximize marketing ROI. In order to handle heterogeneity of different product categories and customer geographies, the team builds a Hierarchical Bayesian logistic regression model. The training data consists of ~261,628 observations (historical offers) distributed across 22 regression groups. We use 50 lower-level covariates and 3 upper-level covariates in the model (including intercepts). The model can be formally described as follows:
\begin{subequations}
\begin{equation}
y_i \sim \mathrm{dBern}(1/(1+\exp(-\xx_i^t \bbeta_{j[i]})))
\end{equation}
\begin{equation}
\bbeta_{j} \sim \mathrm{dNorm}(\zz_j \ggamma, \Sigma)
\end{equation}
%\begin{equation}
%\gamma_{km} \sim \mathrm{dNorm}(0.0, \text{1e+6})
%\end{equation}
%\begin{equation}
%\Sigma_{kk} \sim \mathrm{dIGamma}(\text{1e+6},\text{1e+6})
%\end{equation}
\end{subequations}
where $i$ and $j$ are observation and group indexes, respectively and $\Sigma$ is assumed to be diagonal. There are also non-informative Gaussian and Gamma priors on elements of $\gamma$ and $\Sigma$, respectively.
\subsection{Gibbs Sampling using MH-MGT}
We have three groups of variables to sample from: $\bbeta_j$'s, $\ggamma$, and $\Sigma$. The last two groups have conjugate prior and likelihood functions and can therefore be sampled exactly. For $\bbeta_j$'s, the sampled distribution has two components, log-likelihood and the prior:
\begin{eqnarray}
\log P(\bbeta_j|-) = &&\underbrace{-\sum_{i \in S_j} \left\{ (1-y_i) \: \xx_i \bbeta_j + \log[1+\exp(-\xx_i \bbeta_j)] \right\}}_\text{likelihood} \\ &&\underbrace{-\frac{1}{2} (\bbeta_j - \zz_j \ggamma)^T \Sigma^{-1} (\bbeta_j - \zz_j \ggamma)}_\text{prior}
\end{eqnarray}
where $S_j$ is the set of observation indexes belonging to group $j$. Log-likelihood has a negative definite Hessian according to Appendix \ref{sec:theorem} (base distribution is Bernoulli, for which negative definiteness of Hessian is easy to prove), and the prior has a negative definite Hessian since it is a multivariate Gaussian. The sum of two negative definite matrices is another negative definite matrix. We can therefore use MH-MGT to sample from $\bbeta_j$'s. We choose to perform MH-MGT on blocks of 5 coefficients to improve mixing. For 50 low-level covariates, this implies Gibbs sampling of 10 blocks for each regression group.
\subsection{Results}
Table~\ref{tbl:case-result} shows the results for running 500 iterations of Gibbs sampling (plus 500 burn-in). For comparison, we show same output but using univariate slice sampler, the method of choice used by the modeling team prior to development of MH-MGT technqiue. first, notice the significant heterogeneity across coefficients of different groups. This, in conjunction with the fact that some groups have very few observations (e.g. 55 for one group), justifies the use of an HB framework for partial pooling of the coefficients. We see that the two methods have produced similar coefficient values, but MH-MGT is about 4.1/0.7=5.8x more efficient than slice sampler. (Both methods are implemented in R, and the function and derivative evaluation routines are very similar for the two sampling methods.)

\begin{table*}
\caption{Comparison of MH-MGT and slice sampler results for HB logistic regression case study.}
\label{tbl:case-result}
\begin{tabular}{ccc}
\hline
 & MH-MGT & Slice Sampler \\
\hline
$\beta_{1,1}$ & $1.62\pm0.089$ & $1.65\pm0.13$ \\
\hline
$\beta_{2,1}$ & $-1.43\pm0.43$ & $-1.65\pm0.52$ \\
\hline
$\beta_{3,1}$ & $-1.03\pm0.17$ & $-1.19\pm0.089$ \\
\hline
$\beta_{4,1}$ & $1.27\pm0.17$ & $1.54\pm0.18$ \\
\hline
$\beta_{5,1}$ & $-0.99\pm0.25$ & $-0.89\pm0.14$ \\
\hline
Time for $\bbeta$'s \\ (min) & 110.3 & 750.0 \\
\hline
Average effective size & 155 & 181 \\
\hline
Time per independent sample \\ (min) & 0.7 & 4.1 \\
\hline
\end{tabular}
\end{table*}

\subsection{Sampling Efficiency}
Efficiency of a MCMC sampling technique depends on two factors: 1) calculation needed per sample, and 2) sample autocorrelation. To measure (1), we define `function evaluation equivalents per sample', a number that captures total computation - measured in units of function evluation - to generate each sample, and includes calculation of function, gradient and Hessian, as well as overhead such as solving Newton's equation and performing the MH rejection test. To measure (2), we apply the function \texttt{effectiveSize()} from \texttt{R} package \texttt{coda} and divide it by the number of samples (excluding burn-in). Table \ref{tbl:se-binlogit} summarizes the sampling efficiency analysis of MH-MGT, univariate slice sampler, adaptive rejection sampling (ARS) \cite{bib:gilks-wild} and shrinking-rank slice sampler \cite{bib:mt-shrink-rank-slicer} for binary logistic regression likelihood. Similar to our case study, we see that MH-MGT is $\sim$7x more efficient than univariate slicer and $\sim$3x more efficient than shrinking-rank slicer. Also, contrary to shrinking-rank slicer where the tuning parameter can have significantly different values for different models, MH-MGT requires almost no tuning. (Choosing a chunk size of 5 for MH-MGT appears to be safe across several models that we have tested.)

\begin{table*}
\caption{Sampling efficiency of MH-MGT, univariate slice sampler, ARS, and Shrinking-Rank slice sampler for binary logistic regression. Number of observations: $1000$; number of covariates: $10$. Data is simulated. $100$ runs are used to generate the averages reported. For all reference samplers, we tuned the parameter(s) until no significant improvement in sampling efficienct was seen. Source codes are: univariate slicer (\url{http://www.cs.toronto.edu/~radford/ftp/slice-R-prog}); ARS: \url{http://cran.r-project.org/web/packages/ars/ars.pdf}; shrinking-rank-slicer: \url{http://cran.r-project.org/web/packages/SamplerCompare/index.html}}
\label{tbl:se-binlogit}
\begin{tabular}{lcccc}
\hline \\
 & univariate slicer & ARS & shrinking-rank slicer & MH-MGT\\
\hline
FEE per Nominal Sample & 54.8 & 75.2 & 5.3 & 6.3 \\
\hline
Effective Sampling Rate & 0.83 & 1.10 & 0.21 & 0.70 \\
\hline
FEE per Effective Sample & 69.1 & 72.0 & 26.9 & 9.7 \\
\hline
\end{tabular}
\end{table*}

\section{Discussion} \label{sec:discussion}
\subsection{Improving Convergence and Mixing}
MH-MGT depends on Newton optimization to achieve convergence during burn-in phase. In our case study, this was sufficient. However, Newton optimization can have pathological behavior for some distributions when initial point is too far from the distribution mode. For such cases, we can use modified versions of Newton method, e.g. with line search \cite{bib:nocedal}. In fact, line search can also be incorporated into the proposal function to allow MH-MGT to have better mixing, e.g. for smaller data sizes. Including line search, however, adds to the computational cost of the algorithm, and must therefore be considered carefully before using.
\subsection{Faster Calculation of Hessian}
One barrier for using bigger block partitions in MH-MGT is the quadratic scaling of Hessian calculation with state space dimensionality. Hessian calculation, however, naturally lends itself to parallelization. With parallelization, Hessian calculation time can be reduced significantly, allowing larger partition sizes to become computationally feasible. It is possible that a combination of line search techniques for proposal function, and parallel Hessian calculation can lead to higher sampling efficiency for MH-MGT for high-dimensional problems.
\subsection{Boundary Conditions}
In our example, all coefficients were unconstrained. Including boundary conditions in MH-MGT is straightforward, in principle. Such conditions can be enforced by rejecting out-of-bound draws from the proposal function. However, if the mean of the fitted Gaussian lies close to, or inside, the forbidden region, acceptance rate can be significantly reduced. It remains to be seen whether we need more sophisticated approaches to enforcing boundary conditions for MH-MGT.
\subsection{Beyond Log-concave Distrbutions}
MH-MGT currently works only for distributions where Hessian exists and is negative definite. Depsite ubiquity of such distributions in practice. generalizing MH-MGT beyond current restrictions might be worthwhile. One possible avenue to explore is to use a method such as slice sampler as a backup when calculated Hessian at a point is not negative definite.
\subsection{Conclusions}
The proposed MH-MGT sampling algorithm not only offers potential for significantly faster MCMC, especially for large datasets, it also creates a conceptual bridge to the world of optimization. We hope that future research will further connect these two lines of research, leading to cross-pollination and ultimately more effective sampling as well as optimization techniques.

\appendix

\section{Log-concavity Invariance Theorem} \label{sec:theorem}
\begin{thm}
If $N$ functions $f^i(u_1,...,u_J) : \R^J \rightarrow \R, \: i=1,...,N$ have negative definite Hessians, then the function $g: \R^{\sum_j K_j} \rightarrow \R$ defined as:
\begin{equation} \label{eq:main-theorem}
g(\bbeta^1,...,\bbeta^J) \equiv \sum_{i=1}^N g^i(\bbeta^1,...,\bbeta^J) \equiv \sum_{i=1}^N f^i(\langle \xx_1^i,\bbeta_1 \rangle, ..., \langle \xx_J^i,\bbeta_J \rangle)
\end{equation}
also has a negative definite Hessian IF at least one of $J$ matrices $\XX_1,...,\XX_J$ is full rank, where $\XX_j[N \times K_j] \equiv \begin{bmatrix} \xx_j^1 & ... & \xx_j^N \end{bmatrix}^T$ and $\xx_j^i,\bbeta_j \in \R^{K_j}$.
\end{thm}
\begin{proof}
Applying the chain rule to \eqref{eq:main-theorem}, we can express the Hessian for $g$ in terms of $J^2$ blocks, where block $\HH_{jj'}$ has dimensions $K_j \times K_{j'}$:
\begin{equation} \label{eq:hij}
\HH_{jj'} = \sum_{i=1}^N f_{u_ju_{j'}}^i \: . \: (\xx_j^i \otimes \xx_{j'}^i)
\end{equation}
To prove negative definiteness of $\HH$, we form $ \pp^T \HH \pp$ while decomposing $\pp$ into $J$ subvectors of length $K_j$ each:
\begin{equation} \label{eq:pp}
\pp = \begin{bmatrix} \pp_1^T & \pp_2^T & \cdots & \pp_J^T \end{bmatrix}^T
\end{equation}
We now combine Eqs. \eqref{eq:hij} and \eqref{eq:pp}:
\begin{eqnarray}
\pp^T \HH \pp &=& \sum_{j,j'} \pp_j^T \HH_{jj'} \pp_{j'} \\
&=& \sum_{j,j'} \pp_j^T \left( \sum_i f_{u_ju_{j'}}^i \: . \: (\xx_j^i \otimes \xx_{j'}^i) \right) \pp_{j'} \\
&=& \sum_i \sum_{j,j'} f_{u_ju_{j'}}^i \: \pp_j^T \: (\xx_j^i \otimes \xx_{j'}^i) \: \pp_{j'}
\end{eqnarray}
If we define a set of new vectors $\qq_i$ as:
\begin{eqnarray}
\qq_i \equiv \begin{bmatrix} \pp_1^T \xx_1^i & \cdots & \pp_J^T \xx_J^i \end{bmatrix}
\end{eqnarray}
and use $\HH_i$ to denote the Hessian of $f_i$, we can write:
\begin{equation}
\pp^T \HH \pp = \sum_i \qq_i^T \, \HH_i \, \qq_i
\end{equation}
Since all $\HH_i$'s are negative definite, all $\qq_i^T \, \HH_i \, \qq_i$ terms must be non-positive. Therefore, $\pp^T \HH \pp$ can be non-negative only if all its terms are zero, which is possible only if all $\qq_i$'s are null vectors. This, in turn, means we must have $\pp_j^T \xx_j^i = 0,\:\: \forall \, i,j$. In other words, we must have $\XX_j \pp_j = \emptyset,\:\: \forall \, j$. This means that all $\XX_j$'s have non-singleton nullspaces and therefore cannot be full-rank, which contradicts our assumption. Therefore, $\pp^T \HH \pp$ must be negative.
\end{proof}
Note that we can easily extend the above theorem to include linear basis function models \cite{bib:bishop}.

This theorem provides an absract basis for proving concavity of log-likelihood for many statistical models, including generalized linear family, proportional hazard survival models and multinomial logit regression. However, the framework is more general and can be used to `invent' new regression problems. For example, in Weibull survival analysis we can assume the rate parameter to also be related to a linear function of some covariates (perhaps via a link function), in addition to the hazard function. Since the base, 2-parameter distribution has a negative Hessian, same goes for the expanded distribution, as per the above theorem.

\end{document}